\documentclass[11pt]{article}

\usepackage[hyphens]{url}
\usepackage[pagebackref,colorlinks]{hyperref}
\usepackage[hyphenbreaks]{breakurl}
\usepackage{amsmath} 
\usepackage{amsthm} 
\usepackage{amssymb}	
\usepackage{graphicx} 
\usepackage{multicol} 
\usepackage{multirow}
\usepackage{color}
\usepackage[dvips,letterpaper,margin=1in,bottom=1in]{geometry}
\usepackage[capitalize,noabbrev]{cleveref}

\usepackage{bm}
\usepackage{bbm}

\usepackage{ytableau}

\usepackage[utf8]{inputenc}
\usepackage[english]{babel}

\usepackage[T1]{fontenc}
\AtBeginDocument{%
  \DeclareFontShape{T1}{cmr}{m}{scit}{<->ssub*cmr/m/sc}{}%
}

\usepackage{diagbox}
\usepackage{mathtools}

\newtheorem{theorem}{Theorem}[section]

\newtheorem{lemma}[theorem]{Lemma}

\newtheorem{proposition}[theorem]{Proposition}

\DeclarePairedDelimiter\rbra{\lparen}{\rparen}
\DeclarePairedDelimiter\sbra{\lbrack}{\rbrack}

\DeclarePairedDelimiter\abs{\lvert}{\rvert}

\DeclarePairedDelimiter\ket{\lvert}{\rangle}
\DeclarePairedDelimiter\bra{\langle}{\rvert}

\DeclarePairedDelimiter\norm{\lVert}{\rVert}

\DeclareMathOperator*{\E}{\mathbb{E}}
\DeclareMathOperator*{\Var}{\mathbb{V}\!\!\;ar}

\newcommand{\tr} {\operatorname{tr}}

\newcommand{\supp} {\operatorname{supp}}

\newcommand{\Sym}{\operatorname{Sym}}

\renewcommand{\H} {\operatorname{H}} 
\renewcommand{\S} {\operatorname{S}} 

\newcommand{\Id}{\mathbbm{I}}

\newcommand{\dtr} {\mathrm{d}_{\mathrm{tr}}}

\newcommand{\ketbra}[2]{\ensuremath{\ket{#1}\!\bra{#2}}}

\usepackage{latexsym}
\usepackage{CJK}

\usepackage{enumerate}

\usepackage{algorithm}
\usepackage{algpseudocode}

\usepackage{stmaryrd}
\usepackage{tabularx}
\usepackage{booktabs}
\usepackage{adjustbox}
\usepackage{array}
\newcolumntype{Y}{>{\raggedright\arraybackslash}X}

\newcommand{\footremember}[2]{%
    \footnote{#2}
    \newcounter{#1}
    \setcounter{#1}{\value{footnote}}%
}

\usepackage{tikz}
\usetikzlibrary{quantikz2}

\begin{document}

\title{Optimal von Neumann Entropy Estimation}
\author{Ziyun Chen\footremember{1}{Ziyun Chen is with the Paul G.\ Allen School of Computer Science \& Engineering, University of Washington 
(\href{mailto:ziyuncc@cs.washington.edu}{\nolinkurl{ziyuncc@cs.washington.edu}}).}
\and
Minbo Gao\footremember{2}{Minbo Gao is with the Institute of Software, Chinese Academy of Sciences, and with University of Chinese Academy of Sciences (\href{mailto:gaomb@ios.ac.cn}{\nolinkurl{gaomb@ios.ac.cn}} or
\href{mailto:gmb17@tsinghua.org.cn}{\nolinkurl{gmb17@tsinghua.org.cn}}).}
\and
Qisheng Wang\footremember{3}{Qisheng Wang is with the School of Computer Science, Shanghai Jiao Tong University (\href{mailto:QishengWang1994@gmail.com}{\nolinkurl{QishengWang1994@gmail.com}).}}}
\date{}

\maketitle

\begin{abstract}
We show that the von Neumann entropy of an unknown $d$-dimensional quantum state $\rho$ can be estimated to within additive error $\varepsilon$ using
\[
\Theta\!\left(\frac{d^2}{\varepsilon\log^2(d)(1+\varepsilon\log^2(d))} + \frac{\log^2(d)}{\varepsilon^2}\right)
\]
samples of $\rho$, matching the lower bound due to \hyperlink{cite.Wan26}{Wang (2026)}.
Our estimator is based on a high-degree approximation polynomial and the main technical challenge is bounding its variance despite the large polynomial degree. 
We overcome this using a quantum Hoeffding decomposition, which preserves cancellations between different polynomial terms and reduces the analysis to derivative bounds for the approximation polynomial.
\end{abstract}

\newpage
\tableofcontents
\newpage

\section{Introduction}
\label{sec:introduction}

Entropy estimation is a fundamental problem in information theory. 
In the classical world, the estimation of Shannon entropy \cite{Sha48a,Sha48b} 
\[
\H\rbra{P} = -\sum_{i=1}^d p_i \log\rbra*{p_i}
\]
serves as a key component in quantifying and analyzing sequence variability \cite{SEM91}, neural data \cite{NBvS04}, network traffic \cite{LSO+06}, etc. 
The study of entropy estimation from the perspective of algorithms and complexity dates back to \cite{Pan03}. 
A straightforward approach is to estimate the whole distribution and use a plug-in estimator, which was shown to have sample complexity linear in the dimension \cite{JVHW17}.
It was once an open question whether a sublinear estimator for Shannon entropy exists until the non-constructive sublinear estimator proposed in \cite{Pan04}. 
After that, a series of work \cite{BDKR05,Val11,VV11a,VV11b,VV17,JVHW15,JVHW17,WY16} continued to focus on this direction. 
As a now standard result, it has been fully characterized that, to estimate the Shannon entropy of a $d$-dimensional discrete probability distribution to within additive error $\varepsilon$, it is sufficient and necessary to use $\Theta\rbra{\frac{d}{\varepsilon\log\rbra{d}} + \frac{\log^2\rbra{d}}{\varepsilon^2}}$ samples \cite{JVHW15,WY16}. 

The von Neumann entropy \cite{vN27,vN32}
\[
\S\rbra{\rho} = -\tr\rbra*{\rho \log\rbra{\rho}},
\]
the quantum generalization of Shannon entropy, serves as a key quantum information-theoretic quantity to measure the randomness of quantum systems.
The estimation of von Neumann entropy has applications in entanglement entropy estimation \cite{IMP+15}, quantum Gibbs state preparation \cite{WH19,CLW20,WLW21}, and Hamiltonian learning \cite{AAKS21}. 
The sample complexity of estimating the von Neumann entropy of a $d$-dimensional quantum state to within additive error $\varepsilon$ was first established in \cite{AISW20} as $O\rbra{\frac{d^2}{\varepsilon^2}}$ via a plug-in estimator, or called Empirical Young Diagram (EYD) estimator, based on weak Schur sampling \cite{CHW07}, which was later improved to $O\rbra{\frac{d^2}{\varepsilon}+\frac{\log^2\rbra{d}}{\varepsilon^2}}$ in \cite{BMW16} (cf.\ \cite[Theorem 1.27]{OW17}). 
A time-efficient estimator for von Neumann entropy was given in \cite{WZ25b} with time complexity $\widetilde{O}\rbra{\frac{d^2}{\varepsilon^5}}$. 
A decade later, since all known previous estimators \cite{AISW20,BMW16,WZ25b} have sample complexity $\Omega\rbra{d^2}$ and it has been conjectured (cf.\ \cite{WZ25b,CWZ25}) that $\Omega\rbra{d^2}$ is the sample complexity lower bound for von Neumann entropy estimation, this question still remains:
\[
\textit{Can we estimate the von Neumann entropy with $o(d^2)$ samples?}
\]
A partially negative answer to this question was given in \cite{AISW20}, where they showed that $\Omega\rbra{\frac{d^2}{\varepsilon}}$ samples are necessary to estimate the von Neumann entropy for any plug-in estimator. 

In this paper, in sharp contrast, we give a positive answer to this question by presenting a von Neumann entropy estimator with \textit{optimal} sample complexity that matches the lower bound in \cite{Wan26} and is subquadratic in the dimension $d$,\footnote{The work \cite{FOW26}, concurrent with \cite{Wan26}, also showed an almost matching sample complexity lower bound of $\Omega\rbra{d^{2-\gamma}}$ for von Neumann entropy estimation for any constant $\gamma > 0$.} thereby breaking the $d^2$ barrier for plug-in estimators established in \cite{AISW20} and refuting a conjecture in \cite{WZ25b,CWZ25}.\footnote{In an earlier version of this paper \cite{GW26}, we already presented a von Neumann entropy estimator with  sample complexity $O\rbra{\frac{d^2 \log^2\rbra{\log\rbra{d}} \log\rbra{1/\varepsilon}}{\varepsilon^2 \log^2\rbra{d}} + \frac{\log^2\rbra{d/\varepsilon}}{\varepsilon^2}}$ that is subquadratic in $d$. In the current version, we further improve it to optimal.} 

\begin{theorem} \label{thm:main-intro}
    There is an estimator that, for every $d$-dimensional quantum state $\rho$, estimates the von Neumann entropy $\S\rbra{\rho}$ to within additive error $\varepsilon$ with success probability at least $2/3$, using 
    \[
    \Theta\!\left(\frac{d^2}{\varepsilon\log^2(d)(1+\varepsilon\log^2(d))} + \frac{\log^2(d)}{\varepsilon^2}\right)
    \]
    samples of $\rho$. 
\end{theorem}

The developments in von Neumann entropy estimation are summarized in \cref{tab:von}.

\begin{table}[!htp]
    \centering
    \caption{Sample complexity of von Neumann entropy estimation.}
    \label{tab:von}
    \begin{tabular}{ccc}
    \toprule
    Complexity Type & Sample Complexity & References \\
    \midrule
    \multirow{13}{*}{Upper Bounds} & $O\rbra*{\dfrac{d^2\log^2\rbra{d/\varepsilon}}{\varepsilon^2}}$ & Implied by \cite{OW21} \\
    \addlinespace
    & $O\rbra*{\dfrac{d^2}{\varepsilon^2}}$ & \cite{AISW20} \\
    \addlinespace
    & $O\rbra*{\dfrac{d^2}{\varepsilon} + \dfrac{\log^2\rbra{d}}{\varepsilon^2}}$ & \cite{BMW16} \\
    \addlinespace
    & $O\rbra*{\dfrac{d^2\log^4\rbra{d/\varepsilon}\log^2\rbra{\log\rbra{d}}}{\varepsilon^4\log^2\rbra{d}}}$ & Implied by \cite{PSTW26} \\
    \addlinespace
    & $O\rbra*{\dfrac{d^2\log^2\rbra{d/\varepsilon}}{\varepsilon^2 \log^2\rbra{d}}}$ & Implied by \cite{BST2026} \\
    \addlinespace 
    & $O\!\left(\dfrac{d^2}{\varepsilon\log^2(d)(1+\varepsilon\log^2(d))} + \dfrac{\log^2(d)}{\varepsilon^2}\right)$ & This Work \\
    \midrule
    \multirow{9}{*}{Lower Bounds} & $\Omega\rbra*{\dfrac{d}{\varepsilon\log\rbra{d}}+\dfrac{\log^2\rbra{d}}{\varepsilon^2}}$ & \cite{JVHW15,WY16} \\
    \addlinespace
    & $\Omega\rbra*{\dfrac{d^2}{\varepsilon}}$ (for EYD estimators) & \cite{AISW20} \\
    \addlinespace
    & $\Omega\rbra*{\dfrac{d}{\varepsilon}}$ & \cite{WZ25b} \\
    \addlinespace
    & $\Omega\rbra*{d^{2-\gamma}}$ & \cite{FOW26} \\
    \addlinespace
    & $\Omega\rbra*{ \dfrac{d^2}{\varepsilon\log^2\rbra{d}\rbra{1 + \varepsilon\log^2\rbra{d}}} + \dfrac{\log^2\rbra{d}}{\varepsilon^2} }$ & \cite{Wan26} \\
    \bottomrule
    \end{tabular}
\end{table}

\subsection{Technique Overview}

\paragraph{Prior approaches.}
All prior approaches have sample complexity at least $\Omega\rbra{d^2}$. 

The prior art builds on the Empirical Young Diagram (EYD) algorithm (i.e., the quantum plug-in estimator) in two steps (cf.\ \cite{AISW20}): (i) compute the empirical Schur--Weyl distribution $\widehat{\boldsymbol{\alpha}}$ from $n$ samples of $\rho$ by weak Schur sampling \cite{CHW07}, and (ii) return $\H\rbra{\widehat{\boldsymbol{\alpha}}}$ as the estimate of $\S\rbra{\rho}$.
Under this EYD framework, a von Neumann entropy estimator was given in \cite{AISW20} with sample complexity $O\rbra{\frac{d^2}{\varepsilon^2}}$, which was later improved to $O\rbra{\frac{d^2}{\varepsilon} + \frac{\log^2\rbra{d}}{\varepsilon^2}}$ in \cite{BMW16}. 

Another approach is to apply the spectrum estimation \cite{OW21,PSTW26,BST2026} with the Fannes--Audenaert inequality \cite{Fan73,Aud07}: $\abs*{\S(\rho)-\S(\sigma)}\leq\delta\log(d-1)+\H((\delta,1-\delta))$, where $\delta = \dtr\rbra{\rho, \sigma}$. 
It suffices to set the precision $\delta = \Theta\rbra{\frac{\varepsilon}{\log\rbra{d/\varepsilon}}}$ for the spectrum estimation in the total variation distance to ensure the additive error $\varepsilon$ for von Neumann entropy estimation. 
Using the spectrum estimation with sample complexity $O\rbra{\frac{d^2}{\varepsilon^2}}$ due to \cite{OW21}, $O\rbra{\frac{d^2\log^2\rbra{\log\rbra{d}}}{\varepsilon^4 \log^2\rbra{d}}}$ due to \cite{PSTW26}, and $O\rbra{\frac{d^2}{\varepsilon^2 \log^2\rbra{d} + \varepsilon^4 \log^4\rbra{d}}}$ due to \cite{BST2026} gives a von Neumann entropy estimator with sample complexity $O\rbra{\frac{d^2\log^2\rbra{d/\varepsilon}}{\varepsilon^2}}$, $O\rbra{\frac{d^2\log^4\rbra{d/\varepsilon}\log^2\rbra{\log\rbra{d}}}{\varepsilon^4\log^2\rbra{d}}}$, and $O\rbra{\frac{d^2\log^2\rbra{d/\varepsilon}}{\varepsilon^2 \log^2\rbra{d}}}$, respectively. 

A third approach builds on the Hadamard test \cite{AJL09} based on the samplizer \cite{WZ25a,WZ25b} equipped with quantum singular value transformation \cite{GSLW19}, which gives a von Neumann entropy estimator with sample and time complexity $O\rbra{\frac{d^2}{\varepsilon^5}\log^7\rbra{\frac{d}{\varepsilon}}\log^2\rbra{\frac{\log\rbra{d}}{\varepsilon}}}$ \cite{WZ25b}. 
Although this approach gives a worse sample complexity compared to the aforementioned approaches, it results in a better time complexity. 

\paragraph{The polynomial estimator.}
Our idea is to find an approximation polynomial $p\rbra{x} \approx -x\log\rbra{x}$, which ensures $\tr\rbra{p\rbra{\rho}} \approx \S\rbra{\rho}$. 
To this end, we choose $p \coloneqq p_K$ ass the degree-$K$ polynomial defined by truncating a Chebyshev expansion of $-x\log\rbra{x}$, with the approximation guarantee
\[
\sup_{0\le x\le 1}|p_K(x)+x\log x|
\le \frac{2}{K(K+1)} \qquad \Longrightarrow \qquad \abs*{\tr\rbra{p_K\rbra{\rho}}-\S\rbra{\rho}} \leq \frac{2d}{K\rbra{K+1}}.
\]
Note that taking
$K =\Theta (\sqrt{d/\varepsilon})$
makes the bias $O(\varepsilon)$.
This approximation polynomial was used for entropy estimation in \cite{WBS14}.
Write
$p(x)=\sum_{k=1}^K a_kx^k$.
For each $k$, let $C_k$ be the average of all $k$-cycle permutation operators, and its expectation on $\rho^{\otimes k}$ is
$\operatorname{tr}(\rho^k)$.
Define the $U$-statistic observable
$\mathcal U_{n,k}(C_k)$ which averages $C_k$ over all $k$-subsets of the $n$ registers (cf.\ \cite[Definition VI.3]{GB2010}), and define the polynomial observable 
\[
G_p=\sum_{k=1}^K a_k \mathcal U_{n,k}(C_k).
\]
Let $\widehat F_p$ denote the outcome of measuring $G_p$ on $\rho^{\otimes n}$. Then $
\mathbb{E}[\widehat F_p\mid \rho]
=
\operatorname{tr}(p(\rho))$.
Our von Neumann entropy estimator is defined by
\[
\widehat S \coloneqq \widehat F_{p_K}.
\]
The main difficulty is to bound the variance of the estimator. 
Note that the \textit{quantum Heoffding decomposition} \cite{GB2010} gives a useful identity for the variance: 
\[
\operatorname{Var}[\widehat F_p\mid \rho]
=
\sum_{j=1}^K
\binom{n}{j}^{-1}
\|\mathcal P_jB_{p,j}\|_{\rho,j}^2, \qquad B_{p,j}
=
\sum_{k=j}^K
a_k\binom{k}{j}(C_k)_{\rho\to j},
\]
where $(C_k)_{\rho\to j}$ is obtained by contracting $k-j$ registers against $\rho$, $\mathcal P_j$ is the projection that centers all $j$ remaining registers, and $\|X\|_{\rho,j}^2=
\operatorname{tr}\!\left(\rho^{\otimes j}X^\dagger X\right)$.

\paragraph{Bound the variance by derivatives of $p$.}
For $\rho = \sum_{i} \lambda_i \ketbra{i}{i}$, we observe that $B_{p, j} = D_{p, j} C_j = C_j D_{p, j}$, where
\[
D_{p,j}
=
\sum_{i_1,\ldots,i_j}
\phi_{p,j}(\lambda_{i_1},\ldots,\lambda_{i_j})
\ket{i_1,\ldots,i_j}\!\bra{i_1,\ldots,i_j}, \qquad \phi_{p,j}(x_1,\ldots,x_j)
=
\frac{1}{j!}
\E_t
\left[
p^{(j)}
\left(
\sum_{\ell=1}^j t_\ell x_\ell
\right)
\right].
\]
where $t$ is uniform on the simplex $\{t\in\mathbb R_+^j:\sum_{\ell=1}^j t_\ell=1\}$.
This enables $\norm{B_{p,j}}_{\rho,j}^2
=\tr\rbra{\rho^{\otimes j}B_{p,j}^2}=\tr\rbra{\rho^{\otimes j}D_{p,j}^2C_j^2}$, which can be upper bounded solely by derivatives of the polynomial $p$.
In our case, there exists a constant $C > 0$ such that for $j\ge 2$, the polynomial $p_K$ satisfies
\[
|p_K^{(j)}(x)|
\le
C^jK^{j-2}
\max\left\{
x,\frac{j^2}{K^2}
\right\}^{-j/2},
\]
and the first derivative satisfies
$|p_K'(x)|
\le
C\bigl(1+\log(1/x)\bigr)$.
These together yields
\[
\operatorname{Var}[\widehat S\mid\rho]
\le
\frac{C\log^2 d}{n}
+
\frac{C}{K^4}
\exp\left[
C\min\left\{
\frac{K}{n^{1/4}},
K\sqrt{\frac dn}
\right\}
+
\frac{CK^2}{n}
\right].
\]

Finally, choosing
\[
K = \Theta\rbra*{\sqrt{\frac{d}{\varepsilon}}}, \qquad n=
O\left(
\frac{d^2}
{\varepsilon\log^2 d\,(1+\varepsilon\log^2 d)}
+
\frac{\log^2 d}{\varepsilon^2}
\right),
\]
makes the bias and the standard deviation become $O\rbra{\varepsilon}$. 

\section{Preliminaries}
\label{sec:preliminaries}

For a positive integer $n$, write $[n]=\{1,\ldots,n\}$.
We write $\log\rbra{\cdot}$ to denote the natural logarithm. 
For $j\geq1$, the falling and rising factorials are $x^{\underline j}=x(x-1)\cdots(x-j+1)$ and $x^{\overline j}=x(x+1)\cdots(x+j-1)$.

\subsection{The entropy polynomial}
\label{sec:entropy-polynomial}

Let $T_k$ be the Chebyshev polynomial of the first kind, defined by
$T_k(\cos\theta)=\cos(k\theta)$. For $K\geq2$, define
\begin{equation}
 p_K(x)=\left(\log 2-\frac32\right)x
 -2\sum_{k=2}^{K}\frac{T_k(1-x)-1}{k(k^2-1)}.
 \label{eq:entropy-polynomial}
\end{equation}
This approximation polynomial for $-x\log x$ was used for entropy estimation in \cite[Proposition~3.1]{WBS14}.
The following lemma collects the properties of $p_K$ used later.
\begin{lemma}[Entropy polynomial]
\label{lem:entropy-polynomial}
There is a universal constant $C>0$ such that, for every $K\geq2$:
\begin{align}
 p_K(0)&=0,\label{eq:p-zero}\\
 \sup_{0\leq x\leq1}\abs{p_K(x)+x\log x}
 &\leq \frac{2}{K(K+1)},\label{eq:p-approx}\\
 \abs{p_K'(x)}
 &\leq C\bigl(1+\log(1/x)\bigr),
 &&0<x\leq1,\label{eq:p-first-derivative}\\
 \abs{p_K^{(j)}(x)}
 &\leq C^jK^{j-2}\max\left\{x,\frac{j^2}{K^2}\right\}^{-j/2},
 &&2\leq j\leq K,\ 0\leq x\leq1.
 \label{eq:p-higher-derivative}
\end{align}
\end{lemma}

The proof of this lemma is deferred to \cref{app:entropy-polynomial}.

\subsection{Permutation observables}
\label{sec:permutation-observables}

Let $\Sym\rbra{m}$ be the symmetric group on $[m]$. For $\pi\in\Sym\rbra{m}$, let $U_\pi$ be the permutation operator $U_\pi(\rho_1\otimes\cdots\otimes \rho_m)
=\rho_{\pi^{-1}(1)}\otimes\cdots\otimes \rho_{\pi^{-1}(m)}$.
Let $\mathcal C_k\subseteq\operatorname{Sym}\rbra{k}$ be the set of $k$-cycles, with $\mathcal C_1$ containing only the identity, and define the Hermitian cycle average
\begin{equation}
 C_k=\frac{1}{(k-1)!}\sum_{c\in\mathcal C_k}U_c.
 \label{eq:cycle-average}
\end{equation}
Cyclic permutation measurements of the power trace $\tr(\rho^k)$ are standard in quantum computing; see, e.g., \cite{EAO+02}. For a permutation-invariant operator $h$ on $k$ registers and $n\geq k$, define its $U$-statistic observable, in the sense of quantum $U$-statistics \cite{GB2010}, by
\begin{equation}
 \mathcal U_{n,k}(h)
 =\binom nk^{-1}\sum_{\substack{A\subseteq[n]\\|A|=k}}h^{(A)},
 \label{eq:u-statistic-lift}
\end{equation}
where $h^{(A)}$ acts as $h$ on the registers in $A$ and as the identity on the remaining registers.

For a real polynomial $p(x)=\sum_{k=1}^K a_kx^k$ with $K\leq n$, define
\begin{equation}
 G_p=\sum_{k=1}^K a_k\mathcal U_{n,k}(C_k).
 \label{eq:Gp}
\end{equation}
Since $G_p$ is Hermitian, it can be measured on $\rho^{\otimes n}$. We denote the measurement outcome by $\widehat F_p$.
Namely, let $G_p = \sum_{j} \lambda_j \ketbra{u_j}{u_j}$ be the spectrum decomposition of $G_p$, then $\Pr\sbra{\widehat{F}_p = \lambda_j \mid \rho} = \bra{u_j}\rho^{\otimes n}\ket{u_j}$. 

\begin{proposition}[Unbiased polynomial observable]
\label{prop:unbiased}
For every state $\rho$, $\E\sbra{\widehat F_p\mid\rho}=\tr \rbra{p(\rho)}$.
\end{proposition}

\begin{proof}
For every $k\leq n$,
\[
\tr\rbra*{\rho^{\otimes n}\mathcal U_{n,k}(C_k)}
=\tr\rbra*{\rho^{\otimes k}C_k}
=\tr\rbra{\rho^k},
\]
since every $k$-cycle has the same expectation $\tr\rbra{\rho^k}$ on
$\rho^{\otimes k}$. Hence
$\E\sbra{\widehat F_p\mid\rho}=\sum_{k=1}^K a_k\tr\rbra{\rho^k}
=\tr\rbra{p(\rho)}$.
\end{proof}

\section{The estimator}
\label{sec:estimator}

The estimator is stated explicitly in \cref{alg:entropy-estimator}. 

\begin{algorithm}[H]
\caption{Von Neumann entropy estimator}
\label{alg:entropy-estimator}
\begin{algorithmic}[1]
\Require Integers $n\geq K\geq2$ and $n$ copies of an unknown state $\rho$.
\State Let $p_K(x)=\sum_{k=1}^K a_kx^k$ be the polynomial defined by \eqref{eq:entropy-polynomial}.
\State Let $G_{p_K} = \sum_{k=1}^K a_k \mathcal U_{n,k}(C_k)$, where $C_k$ is defined by \eqref{eq:cycle-average} and $\mathcal U_{n,k}(\cdot)$ is defined by \eqref{eq:u-statistic-lift}.
\State Measure $G_{p_K}$ on $\rho^{\otimes n}$ and let $Y$ be the outcome.
\State \Return $\widehat S\gets Y$.
\end{algorithmic}
\end{algorithm}

Equivalently, in the notation of \cref{sec:permutation-observables}, our estimator is $\widehat S \coloneqq \widehat F_{p_K}$.
The bias of the estimator $\widehat{S}$ is given as follows. 

\begin{proposition}[Bias]
\label{prop:bias}
For every $d$-dimensional state $\rho$,
\begin{equation}
\abs*{\E\sbra{\widehat S\mid\rho}-\S\rbra{\rho}}
\leq \frac{2d}{K(K+1)}.
\label{eq:bias}
\end{equation}
\end{proposition}

\begin{proof}
By Proposition~\ref{prop:unbiased},
$\E\sbra{\widehat S\mid\rho}=\tr \rbra{p_K(\rho)}$.
Applying \eqref{eq:p-approx} in \cref{lem:entropy-polynomial} to the $d$ eigenvalues of $\rho$ gives \eqref{eq:bias}.
\end{proof}

The variance of the estimator $\widehat{S}$ is given as follows. 

\begin{proposition}[Variance]
\label{prop:variance}
There is a constant $C>0$ such that, for every integer $K\geq2$, $n\geq2K$, and every $d$-dimensional quantum state $\rho$, 
\begin{equation}
\Var\sbra{\widehat S\mid\rho}
\leq \frac{C\log^2(d)}{n}
+\frac{C}{K^4}
\exp\!\left[
C\min\!\left\{\frac{K}{n^{1/4}},K\sqrt{\frac dn}\right\}
+\frac{CK^2}{n}
\right].
\label{eq:variance-main}
\end{equation}
\end{proposition}

\section{Variance analysis}
\label{sec:variance-analysis}

We prove Proposition~\ref{prop:variance} in three steps. First, we apply the quantum $U$-statistic variance decomposition of \cite{GB2010}, reducing the variance to norms of effective $j$-register kernels. Second, we identify those kernels explicitly through a simplex average of scalar derivatives. Third, Lemma~\ref{lem:entropy-polynomial} bounds those derivatives for $p_K$.

\subsection{Hoeffding reduction via quantum \texorpdfstring{$U$}{U}-statistics}
\label{sec:hoeffding}

Fix a state $\rho$ and, for Hermitian $X$ on $m$ registers, write
\begin{equation}
\norm{X}_{\rho,m}^2=\tr\rbra*{\rho^{\otimes m}X^2}.
\label{eq:rho-weighted-norm}
\end{equation}
Let $\mathcal P_j$ denote the canonical order-$j$ Hoeffding projection
$P_{\{1,\ldots,j\}}$ of Gu\c{t}\u{a} and Butucea~\cite[Theorem~5.6]{GB2010}.
Their conditional expectation is the orthogonal projection onto the corresponding
coordinate subspace~\cite[Definition~5.1 and Lemma~5.2]{GB2010}; hence
\begin{equation}
\norm{\mathcal P_jX}_{\rho,j}\leq\norm{X}_{\rho,j}.
\label{eq:hoeffding-contraction}
\end{equation}
The same formulas apply to singular $\rho$ after restricting all registers to
$\supp\rbra{\rho}$, which leaves the $\rho$-weighted quantities unchanged.

For a permutation-invariant operator $h$ on $k$ registers and $j\leq k$, define
its contraction to the first $j$ registers by
\begin{equation}
 h_{\rho\to j}
 =\tr_{j+1:k}\!\left[h\rbra*{\Id^{\otimes j}\otimes\rho^{\otimes(k-j)}}\right].
 \label{eq:contraction}
\end{equation}
For $p(x)=\sum_{k=1}^K a_kx^k$, set
\begin{equation}
 B_{p,j}=\sum_{k=j}^K a_k\binom{k}{j}(C_k)_{\rho\to j},
 \qquad 1\leq j\leq K.
 \label{eq:Bpj}
\end{equation}

\begin{lemma}[Hoeffding variance reduction]
\label{lem:hoeffding}
Let $K\leq n$. For every polynomial $p$ and state $\rho$,
\begin{equation}
\Var\sbra{\widehat F_p\mid\rho}
=\sum_{j=1}^K\binom nj^{-1}
\norm{\mathcal P_jB_{p,j}}_{\rho,j}^2
\leq
\sum_{j=1}^K\binom nj^{-1}\norm{B_{p,j}}_{\rho,j}^2.
\label{eq:hoeffding-variance}
\end{equation}
\end{lemma}

\begin{proof}
Define the symmetric $K$-register kernel
\begin{equation}
H_p=\sum_{k=1}^K a_k\mathcal U_{K,k}(C_k).
\label{eq:Hp}
\end{equation}
It is a quantum $U$-statistic kernel in the sense of
\cite[Definition~6.3]{GB2010}. Double counting gives
\begin{align}
\mathcal U_{n,K}(H_p)
&=\sum_{k=1}^K a_k\mathcal U_{n,K}\!\rbra*{\mathcal U_{K,k}(C_k)}\notag\\
&=\sum_{k=1}^K a_k
\frac{\binom{n-k}{K-k}}{\binom nK\binom Kk}
\sum_{\substack{A\subseteq[n]\\|A|=k}}C_k^{(A)}\notag\\
&=\sum_{k=1}^K a_k\mathcal U_{n,k}(C_k)
=G_p,
\label{eq:Gp-as-one-Ustat}
\end{align}
where
$\binom nK\binom Kk=\binom nk\binom{n-k}{K-k}$.

Let $(H_p)_j=\mathcal P_jH_p$ be its order-$j$ Hoeffding kernel.
For $k\geq j$, the conditional-expectation formula of
\cite[Definition~5.1 and Theorem~5.6]{GB2010} and a count of the
$k$-subsets containing $[j]$ give
\begin{align}
\mathcal P_j\mathcal U_{K,k}(C_k)
&=\frac{\binom{K-j}{k-j}}{\binom Kk}\,
\mathcal P_j(C_k)_{\rho\to j}\notag\\
&=\frac{\binom kj}{\binom Kj}\,
\mathcal P_j(C_k)_{\rho\to j}.
\label{eq:project-lift}
\end{align}
For $k<j$ this projection is zero. Therefore, by \eqref{eq:Bpj},
\begin{equation}
(H_p)_j=\binom Kj^{-1}\mathcal P_jB_{p,j}.
\label{eq:Hp-j-kernel}
\end{equation}
Applying the exact variance decomposition
\cite[Lemma~6.4]{GB2010} to \eqref{eq:Gp-as-one-Ustat} yields
\begin{align*}
\Var\sbra{\widehat F_p\mid\rho}
&=\sum_{j=1}^K\binom Kj^2\binom nj^{-1}
\norm{(H_p)_j}_{\rho,j}^2\\
&=\sum_{j=1}^K\binom nj^{-1}
\norm{\mathcal P_jB_{p,j}}_{\rho,j}^2.
\end{align*}
The inequality follows from \eqref{eq:hoeffding-contraction}.
\end{proof}

\subsection{Kernel representation and norm bound}
\label{sec:kernel-analysis}

Write $\rho=\sum_{i=1}^d\lambda_i\lvert i\rangle\!\langle i\rvert$.
For $j\geq1$, let $t=(t_1,\ldots,t_j)$ be uniformly distributed on the simplex $\Delta_{j-1}=\{t\in\mathbb R_+^j:\sum_{\ell=1}^j t_\ell=1\}$,
with $t_1=1$ when $j=1$, and define
\begin{equation}
\phi_{p,j}(x_1,\ldots,x_j)
=\frac{1}{j!}\E_t\sbra*{p^{(j)}\!\left(\sum_{\ell=1}^j t_\ell x_\ell\right)}.
\label{eq:phi}
\end{equation}

\begin{lemma}[Kernel representation]
\label{lem:kernel-representation}
For $1\leq j\leq K$, define the diagonal operator
\begin{equation}
D_{p,j}
=\sum_{i_1,\ldots,i_j=1}^d
\phi_{p,j}(\lambda_{i_1},\ldots,\lambda_{i_j})
\lvert i_1,\ldots,i_j\rangle\!\langle i_1,\ldots,i_j\rvert.
\label{eq:Dpj}
\end{equation}
Then
\begin{equation}
 B_{p,j}=D_{p,j}C_j=C_jD_{p,j}.
 \label{eq:Bpj-factorization}
\end{equation}
\end{lemma}

\begin{proof}
Both sides of \eqref{eq:Bpj-factorization} depend linearly on $p$, so it is enough to prove the claim for a monomial $p(x)=x^k$ with $k\geq j$.
For this monomial,
$p^{(j)}(x)=k^{\underline j}x^{k-j}$.
Expanding the power in \eqref{eq:phi}, using the standard mixed-moment formula for the Dirichlet distribution \cite[Section~2.1.1]{NgTianTang2011}  and choosing $t\sim\mathrm{Dirichlet}(1,\ldots,1)$, we get
\[
\E_t\sbra*{\prod_{\ell=1}^j t_\ell^{r_\ell}}
=\frac{(j-1)!\prod_{\ell=1}^j r_\ell!}{(k-1)!},
\qquad
r_1+\cdots+r_j=k-j,
\]
which implies
\begin{align*}
\phi_{x^k,j}(x_1,\ldots,x_j)
&=\frac{k^{\underline j}}{j!}
  \sum_{\substack{r_1+\cdots+r_j=k-j\\r_\ell\geq0}}
  \binom{k-j}{r_1,\ldots,r_j}
  \E_t\sbra*{\prod_{\ell=1}^j t_\ell^{r_\ell}}
  \prod_{\ell=1}^j x_\ell^{r_\ell}\\
&=\frac{1}{j!}\frac{k!}{(k-j)!}
  \sum_{r_1+\cdots+r_j=k-j}
  \frac{(k-j)!}{\prod_\ell r_\ell!}
  \frac{(j-1)!\prod_\ell r_\ell!}{(k-1)!}
  \prod_{\ell=1}^j x_\ell^{r_\ell}\\
&=\frac{k}{j}
\sum_{\substack{r_1+\cdots+r_j=k-j\\r_\ell\geq0}}
\prod_{\ell=1}^j x_\ell^{r_\ell}.
\end{align*}
Thus
\begin{equation}
\phi_{x^k,j}(x_1,\ldots,x_j)
=\frac{k}{j}
\sum_{\substack{r_1+\cdots+r_j=k-j\\r_\ell\geq0}}
\prod_{\ell=1}^j x_\ell^{r_\ell}.
\label{eq:phi-monomial}
\end{equation}
We now compute the contraction of $C_k$. Let
\[
R := [k]\setminus [j]
\]
be the set of the $k-j$ labels that will be contracted out.

Fix a $j$-cycle $c\in C_j$. We view a $k$-cycle
$\widetilde c\in C_k$ extending $c$ as being obtained by inserting
the labels in $R$ into the $j$ directed edges of $c$.
For each $\ell\in [j]$, let $r_\ell\ge 0$ denote the number of
inserted labels on the edge $c^{-1}(\ell)\longrightarrow \ell$.
For a fixed vector
$r=(r_1,\ldots,r_j)\in \mathbb{N}_0^j$,
with $r_1+\cdots+r_j=k-j$,
write the labels inserted on this edge, in their order along the
cycle, as $w_{\ell,1},\ldots,w_{\ell,r_\ell}$.
Then, in the extended cycle $\widetilde c$, the edge
$c^{-1}(\ell)\to \ell$ is replaced by the directed path
$c^{-1}(\ell)
\longrightarrow
w_{\ell,1}
\longrightarrow \cdots \longrightarrow
w_{\ell,r_\ell}
\longrightarrow
\ell$,
with the convention that when $r_\ell=0$, this path is simply
$c^{-1}(\ell)\to \ell$.

Hence, once $c$ and $r$ are fixed, the cycle $\widetilde c$ is
determined by assigning the $k-j$ labels in $R$ bijectively to the
ordered slots
$\{(\ell,q): 1\le \ell\le j,\ 1\le q\le r_\ell\}$.
There are exactly $(k-j)!$ such assignments.

Conversely, every $k$-cycle $\widetilde c$ determines uniquely:
the $j$-cycle $c$ obtained by deleting the labels in $R$,
the vector $r=(r_1,\ldots,r_j)$ recording the number of deleted
labels in each gap, and the ordering of the deleted labels within
those gaps. Thus we have the bijective correspondence
\[
\widetilde c
\quad\longleftrightarrow\quad
\bigl(c,r,\text{ordering of the labels in }R\bigr).
\]

For such an extension, use the eigenbasis of $\rho$ and the expansion
\[
U_{\widetilde c}
=\sum_{a_1,\ldots,a_k=1}^d
\bigotimes_{v=1}^k
\lvert a_v\rangle\!\langle a_{\widetilde c(v)}\rvert.
\]
Consider the gap ending at $\ell$. If $r_\ell\geq1$, contracting its deleted registers gives
\begin{align*}
&\tr_{w_{\ell,1},\ldots,w_{\ell,r_\ell}}
\left[
\left(
\lvert a_{c^{-1}(\ell)}\rangle\!\langle a_{w_{\ell,1}}\rvert
\otimes
\bigotimes_{q=1}^{r_\ell-1}
\lvert a_{w_{\ell,q}}\rangle\!\langle a_{w_{\ell,q+1}}\rvert
\otimes
\lvert a_{w_{\ell,r_\ell}}\rangle\!\langle a_\ell\rvert
\right)
\left(\Id\otimes\rho^{\otimes r_\ell}\right)
\right]\\
&\qquad=
\left[
\prod_{q=1}^{r_\ell-1}
\langle a_{w_{\ell,q+1}}|\rho|a_{w_{\ell,q}}\rangle
\right]
\langle a_\ell|\rho|a_{w_{\ell,r_\ell}}\rangle
\lvert a_{c^{-1}(\ell)}\rangle\!\langle a_{w_{\ell,1}}\rvert\\
&\qquad=
\lambda_{a_\ell}^{r_\ell}
\lvert a_{c^{-1}(\ell)}\rangle\!\langle a_\ell\rvert,
\end{align*}
where the Kronecker deltas from the diagonal matrix $\rho$ force
$a_{w_{\ell,1}}=\cdots=a_{w_{\ell,r_\ell}}=a_\ell$.
The same formula with factor $1$ holds when $r_\ell=0$. Multiplying the $j$ gap contributions and renaming the retained indices gives
\begin{equation*}
(U_{\widetilde c})_{\rho\to j}
=\sum_{i_1,\ldots,i_j=1}^d
\left(\prod_{\ell=1}^j\lambda_{i_\ell}^{r_\ell}\right)
\bigotimes_{\ell=1}^j
\lvert i_\ell\rangle\!\langle i_{c(\ell)}\rvert.
\end{equation*}
Therefore, summing first over the $(k-j)!$ extensions for each pair $(c,\mathbf r)$ and then averaging over all $(k-1)!$ $k$-cycles,
\begin{align*}
(C_k)_{\rho\to j}
&=\frac{1}{(k-1)!}
\sum_{\widetilde c\in\mathcal C_k}
(U_{\widetilde c})_{\rho\to j}\\
&=\frac{(k-j)!}{(k-1)!}
\sum_{c\in\mathcal C_j}
\sum_{\substack{r_1+\cdots+r_j=k-j\\r_\ell\geq0}}
\sum_{i_1,\ldots,i_j=1}^d
\left(\prod_{\ell=1}^j\lambda_{i_\ell}^{r_\ell}\right)
\bigotimes_{\ell=1}^j
\lvert i_\ell\rangle\!\langle i_{c(\ell)}\rvert\\
&=\frac{(k-j)!}{(k-1)!}
\sum_{c\in\mathcal C_j}
\sum_{i_1,\ldots,i_j=1}^d
\left(
\sum_{\substack{r_1+\cdots+r_j=k-j\\r_\ell\geq0}}
\prod_{\ell=1}^j\lambda_{i_\ell}^{r_\ell}
\right)
\bigotimes_{\ell=1}^j
\lvert i_\ell\rangle\!\langle i_{c(\ell)}\rvert.
\end{align*}
Multiplying by the coefficient $\binom{k}{j}$ appearing in \eqref{eq:Bpj} gives
\begin{align*}
\binom{k}{j}(C_k)_{\rho\to j}
&=\frac{k}{j!}
\sum_{c\in\mathcal C_j}
\sum_{i_1,\ldots,i_j=1}^d
\left(
\sum_{r_1+\cdots+r_j=k-j}
\prod_{\ell=1}^j\lambda_{i_\ell}^{r_\ell}
\right)
\bigotimes_{\ell=1}^j
\lvert i_\ell\rangle\!\langle i_{c(\ell)}\rvert,
\end{align*}
where we used
$\binom{k}{j}(k-j)!/(k-1)!=k/j!$.
On the other hand, by \eqref{eq:cycle-average} and \eqref{eq:phi-monomial},
\begin{align*}
D_{x^k,j}C_j
&=\frac{1}{(j-1)!}
\sum_{c\in\mathcal C_j}
\sum_{i_1,\ldots,i_j=1}^d
\phi_{x^k,j}(\lambda_{i_1},\ldots,\lambda_{i_j})
\bigotimes_{\ell=1}^j
\lvert i_\ell\rangle\!\langle i_{c(\ell)}\rvert\\
&=\frac{k}{j!}
\sum_{c\in\mathcal C_j}
\sum_{i_1,\ldots,i_j=1}^d
\left(
\sum_{r_1+\cdots+r_j=k-j}
\prod_{\ell=1}^j\lambda_{i_\ell}^{r_\ell}
\right)
\bigotimes_{\ell=1}^j
\lvert i_\ell\rangle\!\langle i_{c(\ell)}\rvert.
\end{align*}
In other words, the contribution of $x^k$ to $B_{p,j}$ equals $D_{x^k,j}C_j$. Summing over the monomials of $p$ gives $B_{p,j}=D_{p,j}C_j$.
Finally, $\phi_{p,j}$ is symmetric in its $j$ arguments. Hence $D_{p,j}$ is invariant under conjugation by every register permutation, so it commutes with every $U_\pi$, and in particular with their average $C_j$. This proves the second equality in \eqref{eq:Bpj-factorization}.
\end{proof}

The next lemma converts a weighted derivative bound for $p$ into a norm bound for the corresponding kernel $B_{p,j}$.

\begin{lemma}[Kernel norm from derivative control]
\label{lem:kernel-norm}
Suppose that, for some $\delta,M>0$,
\begin{equation}
\abs{p^{(j)}(x)}
\leq M\max\{x,\delta\}^{-j/2},
\qquad 0\leq x\leq1.
\label{eq:derivative-hypothesis}
\end{equation}
Set
\begin{equation}
s_\delta(\rho)
=\sum_{i=1}^d\frac{\lambda_i}{\max\{\lambda_i,\delta\}}
\leq \min\{d,1/\delta\}.
\label{eq:effective-dimension-general}
\end{equation}
Then
\begin{equation}
\norm{B_{p,j}}_{\rho,j}^2
\leq
\frac{2^jM^2}{(j!)^2(j-1)!}\,
 s_\delta(\rho)^{\overline j}.
\label{eq:kernel-norm-bound}
\end{equation}
\end{lemma}

\begin{proof}
Fix $x_1,\ldots,x_j\in[0,1]$ and write
$z_\ell=\max\{x_\ell,\delta\}$.
For $t$ uniform on $\Delta_{j-1}$,
\begin{equation}
\E_t\sbra*{\left(\sum_{\ell=1}^j t_\ell z_\ell\right)^{-j}}
=\prod_{\ell=1}^jz_\ell^{-1}.
\label{eq:simplex-negative-moment}
\end{equation}
Identity \eqref{eq:simplex-negative-moment} is exactly the Feynman-parameter identity with all denominator exponents equal to $1$; see \cite[Eq.~(39)]{KCL23}. 
Let
$\mu=\sum_{\ell=1}^jt_\ell x_\ell$ and
$a=\sum_{\ell=1}^jt_\ell z_\ell$.
Since $z_\ell\leq x_\ell+\delta$,
\[
a\leq\mu+\delta\leq2\max\{\mu,\delta\}.
\]
Using \eqref{eq:phi}, the derivative hypothesis, the last inequality, Cauchy--Schwarz, and finally \eqref{eq:simplex-negative-moment}, we have
\begin{equation}
\begin{aligned}
\abs{\phi_{p,j}(x_1,\ldots,x_j)}
&\leq \frac{M}{j!}
   \E_t\sbra*{\max\{\mu,\delta\}^{-j/2}}\\
&\leq \frac{2^{j/2}M}{j!}
   \E_t\sbra*{a^{-j/2}}\\
&\leq \frac{2^{j/2}M}{j!}
   \left(\E_t\sbra*{a^{-j}}\right)^{1/2}\\
&=\frac{2^{j/2}M}{j!}
  \prod_{\ell=1}^jz_\ell^{-1/2},
\end{aligned}
\label{eq:phi-weighted-bound}
\end{equation}
Hence,
\[\abs{\phi_{p,j}(x_1,\ldots,x_j)}
\prod_{\ell=1}^j\sqrt{\max\{x_\ell,\delta\}}
\leq \frac{2^{j/2}M}{j!}.\]

Define the local diagonal operator
\[
Q_\delta
=\sum_{i=1}^d
\frac{\lambda_i}{\max\{\lambda_i,\delta\}}
\lvert i\rangle\!\langle i\rvert.
\]
Then $0\preceq Q_\delta\preceq\Id$ and $\tr Q_\delta=s_\delta(\rho)$.
Squaring \eqref{eq:phi-weighted-bound} and evaluating it on every diagonal basis vector $(i_1,\ldots,i_j)$ gives
\[
\left(\prod_{\ell=1}^j\lambda_{i_\ell}\right)
\abs{\phi_{p,j}(\lambda_{i_1},\ldots,\lambda_{i_j})}^2
\leq
\frac{2^jM^2}{(j!)^2}
\prod_{\ell=1}^j
\frac{\lambda_{i_\ell}}{\max\{\lambda_{i_\ell},\delta\}}.
\]
Equivalently, as an operator inequality,
\[
0\preceq \rho^{\otimes j}D_{p,j}^2
\preceq \frac{2^jM^2}{(j!)^2}Q_\delta^{\otimes j}.
\]
By Lemma~\ref{lem:kernel-representation}, $B_{p,j}=D_{p,j}C_j=C_jD_{p,j}$, and hence
\begin{equation}
\begin{aligned}
\norm{B_{p,j}}_{\rho,j}^2
&=\tr\!\left(\rho^{\otimes j}B_{p,j}^2\right)\\
&=\tr\!\left(\rho^{\otimes j}D_{p,j}^2C_j^2\right)\\
&\leq \frac{2^jM^2}{(j!)^2}
\tr\!\left(Q_\delta^{\otimes j}C_j^2\right).
\end{aligned}
\label{eq:kernel-reduction}
\end{equation}

It remains to bound the last trace. For every $\pi\in\operatorname{Sym}\rbra{j}$, the cycle-trace identity gives
\[
\tr(Q_\delta^{\otimes j}U_\pi)
=\prod_{\gamma\text{ cycle of }\pi}\tr(Q_\delta^{|\gamma|}).
\]
All eigenvalues of $Q_\delta$ lie in $[0,1]$, so
$0\leq\tr(Q_\delta^m)\leq\tr Q_\delta=s_\delta(\rho)$ for every $m\geq1$. Therefore
\[
0\leq\tr(Q_\delta^{\otimes j}U_\pi)
\leq s_\delta(\rho)^{\#\mathrm{cycles}(\pi)}.
\]
Since the set of $j$-cycles is closed under inversion,
\begin{align*}
\tr(Q_\delta^{\otimes j}C_j^2)
&=\frac{1}{((j-1)!)^2}
\sum_{c,c'\in\mathcal C_j}
\tr\!\left(Q_\delta^{\otimes j}U_{c^{-1}c'}\right).
\end{align*}
For each fixed $c$, the permutations $c^{-1}c'$ are distinct as $c'$ ranges over $\mathcal C_j$; hence their sum is at most the sum over all $\pi\in\operatorname{Sym}\rbra{j}$. Using $|\mathcal C_j|=(j-1)!$,
\begin{align*}
\tr(Q_\delta^{\otimes j}C_j^2)
&\leq \frac{1}{(j-1)!}
\sum_{\pi\in\operatorname{Sym}\rbra{j}}s_\delta(\rho)^{\#\mathrm{cycles}(\pi)}\\
&=\frac{s_\delta(\rho)^{\overline j}}{(j-1)!},
\end{align*}
where the last equality, $\sum_{\pi\in\operatorname{Sym}\rbra{j}}s^{\#\mathrm{cycles}(\pi)}=s^{\overline j}$,  is the standard generating identity for unsigned Stirling numbers of the first kind (see \cite[Section~6.1, Equation~(6.11)]{GKP94}).
Substituting this bound into \eqref{eq:kernel-reduction} proves \eqref{eq:kernel-norm-bound}.
\end{proof}

\subsection{Proof of Proposition~\ref{prop:variance}}
\label{sec:variance-proof}

We now prove Proposition~\ref{prop:variance} by combining the Hoeffding decomposition with the kernel bound above.

\begin{proof}[Proof of Proposition~\ref{prop:variance}]
We set $p=p_K$ in Lemmas~\ref{lem:hoeffding} and~\ref{lem:kernel-norm}.

For $j=1$, \eqref{eq:phi} gives $\phi_{p_K,1}(x)=p_K'(x)$ and $C_1=\Id$. Hence Lemma~\ref{lem:kernel-representation} yields
$B_{p_K,1}=p_K'(\rho)$. Since $\mathcal P_1$ is an orthogonal projection,
\[
\norm{\mathcal P_1B_{p_K,1}}_{\rho,1}^2
\leq \norm{p_K'(\rho)}_{\rho,1}^2
=\sum_{\lambda_i>0}\lambda_i\abs{p_K'(\lambda_i)}^2.
\]
By \eqref{eq:p-first-derivative}, the last quantity is at most $C\log^2(d)$. 

In fact, sample $I$ according to $\Pr[I=i]=\lambda_i$ and set $Z=\log(1/\lambda_I)$. Then
\[
\Pr[Z>t]
=\sum_{i:\,\lambda_i<e^{-t}}\lambda_i
\leq \min\{1,de^{-t}\}.
\]
The tail-integral formula therefore gives
\begin{align*}
\E[Z^2]
&=2\int_0^\infty t\Pr[Z>t] \,\mathrm{d}t\\
&\leq 2\int_0^{\log(d)}t\,\mathrm{d}t
 +2d\int_{\log(d)}^\infty te^{-t}\,\mathrm{d}t\\
&=\log^2(d)+2\log(d)+2
\leq C\log^2(d),
\end{align*}
where the last inequality uses $d\geq2$. Since
$\abs{p_K'(\lambda_I)}\leq C(1+Z)$, this proves the claimed $C\log^2(d)$ bound. Thus the order-one contribution to \eqref{eq:hoeffding-variance} is at most $C\log^2(d)/n$.

Now fix $2\leq j\leq K$ and set
\begin{equation}
\delta_j=\frac{j^2}{K^2},
\qquad
s_j=s_{\delta_j}(\rho)
\leq\min\left\{d,\frac{K^2}{j^2}\right\}.
\label{eq:sj}
\end{equation}
By \eqref{eq:p-higher-derivative}, the hypothesis of Lemma~\ref{lem:kernel-norm} holds with
$M=C^jK^{j-2}$. Hence
\[
\norm{B_{p_K,j}}_{\rho,j}^2
\leq
\frac{C^{2j}K^{2j-4}s_j^{\overline j}}
{(j!)^2(j-1)!},
\]
where the factor $2^j$ from Lemma~\ref{lem:kernel-norm} has been absorbed into $C^{2j}$.
Since $n\geq2K\geq2j$,
\[
n^{\underline j}
=n(n-1)\cdots(n-j+1)
\geq (n-j+1)^j
\geq (n/2)^j.
\]
Using $\binom nj^{-1}=j!/n^{\underline j}$, we obtain
\begin{equation}
\begin{aligned}
\binom nj^{-1}\norm{B_{p_K,j}}_{\rho,j}^2
&\leq
\frac{j!}{(n/2)^j}
\frac{C^{2j}K^{2j-4}s_j^{\overline j}}
{(j!)^2(j-1)!}\\
&=\frac{1}{K^4}
\frac{(CK^2/n)^j s_j^{\overline j}}
{j!(j-1)!}.
\end{aligned}
\label{eq:j-contribution}
\end{equation}
Combining the $j=1$ estimate, \eqref{eq:hoeffding-variance}, and \eqref{eq:j-contribution} gives
\begin{equation}
\Var\sbra{\widehat F_{p_K}\mid\rho}
\leq \frac{C\log^2(d)}{n}
+\frac{1}{K^4}
\sum_{j=2}^K
\frac{(CK^2/n)^j s_j^{\overline j}}
{j!(j-1)!}.
\label{eq:variance-series}
\end{equation}

It remains to sum the series. Put $z=CK^2/n$ for sufficiently large $C$. For $s\geq0$,
\[
s^{\overline j}=s(s+1)\cdots(s+j-1)
\leq(s+j)^j
\leq2^{j-1}(s^j+j^j).
\]
By \eqref{eq:sj}, the term involving $s_j^j$ can be bounded either by $d^j$ or by $(K^2/j^2)^j$. Extending the finite sums to infinite sums and absorbing $2^{j-1}$ into $C^j$, the second term in \eqref{eq:variance-series} is at most
\begin{equation}
\frac{C}{K^4}\left[
\min\left\{
\sum_{j\geq2}\frac{(zd)^j}{j!(j-1)!},
\sum_{j\geq2}\frac{(zK^2)^j}{j^{2j}j!(j-1)!}
\right\}
+
\sum_{j\geq2}\frac{(zj)^j}{j!(j-1)!}
\right].
\label{eq:three-series}
\end{equation}

We use the elementary inequality
\[
\sum_{j\geq0}\frac{u^j}{(j!)^q}
\leq \exp\!\left(qu^{1/q}\right),
\qquad u\geq0,\ q\in\mathbb N,
\]
which follows because the left-hand side is a subseries of
$(\sum_{j\geq0}u^{j/q}/j!)^q=e^{qu^{1/q}}$.
The three series in \eqref{eq:three-series} are bounded as follows. First, since
$j!(j-1)!=(j!)^2/j$ and $j\leq2^j$,
\[
\sum_{j\geq2}\frac{u^j}{j!(j-1)!}
=\sum_{j\geq2}\frac{j u^j}{(j!)^2}
\leq\sum_{j\geq0}\frac{(2u)^j}{(j!)^2}
\leq\exp(C\sqrt u).
\]
Second, using $j^{2j}\geq(j!)^2$ and again $j\leq2^j$,
\[
\sum_{j\geq2}\frac{v^j}{j^{2j}j!(j-1)!}
\leq\sum_{j\geq2}\frac{jv^j}{(j!)^4}
\leq\sum_{j\geq0}\frac{(2v)^j}{(j!)^4}
\leq\exp(Cv^{1/4}).
\]
Finally, $j!\geq(j/e)^j$ implies $j^j/j!\leq e^j$, and therefore
\[
\sum_{j\geq2}\frac{(zj)^j}{j!(j-1)!}
=\sum_{j\geq2}\frac{j(zj)^j}{(j!)^2}
\leq\sum_{j\geq2}\frac{j(ez)^j}{j!}
\leq\sum_{j\geq0}\frac{(2ez)^j}{j!}
\leq\exp(Cz).
\]
Applying these bounds in \eqref{eq:three-series} with
$u=zd$ and $v=zK^2$ gives
\[
\frac{C}{K^4}
\exp\!\left[
C\min\left\{\sqrt{zd},(zK^2)^{1/4}\right\}+Cz
\right].
\]
Since $z=CK^2/n$, after enlarging $C$ once more this becomes
\[
\frac{C}{K^4}
\exp\!\left[
C\min\left\{K\sqrt{\frac dn},\frac{K}{n^{1/4}}\right\}
+\frac{CK^2}{n}
\right].
\]
Together with the order-one contribution, this is exactly \eqref{eq:variance-main}.
\end{proof}

\section{Proof of the main theorem}
\label{sec:main-proof}

\begin{proof}[Proof of \cref{thm:main-intro}]
Let $C, C_1, C_2$ be sufficiently large constant chosen in this order.

Set
\begin{equation}
K=\left\lceil\sqrt{\frac{C_1d}{\varepsilon}}\right\rceil,
\qquad
n=\left\lceil
C_2\left(
\frac{d^2}{\varepsilon\log^2(d)\bigl(1+\varepsilon\log^2(d)\bigr)}
+\frac{\log^2(d)}{\varepsilon^2}
\right)
\right\rceil.
\label{eq:parameter-choice}
\end{equation}
For $d\geq2$, $0<\varepsilon\leq1$, we can choose $C_1$ such that
\begin{equation}
\sqrt{\frac{C_1d}{\varepsilon}}
\leq K
\leq 2\sqrt{\frac{C_1d}{\varepsilon}}.
\label{eq:K-bounds}
\end{equation}
We also have
\[
\frac{K^2}{n}
\leq
\frac{CC_1}{C_2}
\frac{\log^2(d)\bigl(1+\varepsilon\log^2(d)\bigr)}{d}
\leq 1,
\]
where the last inequality uses $\varepsilon\leq1$ and choosing $C_2$ large enough. Thus $n\geq K^2$. Since $K\geq4$, we also have $n\geq2K$, so Proposition~\ref{prop:variance} applies.

We next bound the exponent in \eqref{eq:variance-main}. There are two regimes.

If $\varepsilon\log^2(d)\geq1$, then
$1+\varepsilon\log^2(d)\leq2\varepsilon\log^2(d)$, and the first term in \eqref{eq:parameter-choice} gives
\[
n\geq
\frac{C_2d^2}{2\varepsilon^2\log^4(d)}.
\]
Combining this with the upper bound on $K$ in \eqref{eq:K-bounds},
\[
\frac{K}{n^{1/4}}
\leq C\sqrt{C_1}\,C_2^{-1/4}\log(d)
\leq\frac18\log(d).
\]
If $\varepsilon\log^2(d)\leq1$, then
$1+\varepsilon\log^2(d)\leq2$, so
\[
n\geq
\frac{C_2d^2}{2\varepsilon\log^2(d)}.
\]
Again using \eqref{eq:K-bounds},
\[
K\sqrt{\frac dn}
\leq C\sqrt{C_1}\,C_2^{-1/2}\log(d)
\leq\frac18\log(d).
\]
Thus in both regimes
\[
C\min\left\{\frac{K}{n^{1/4}},K\sqrt{\frac dn}\right\}
\leq\frac18\log(d).
\]
Moreover,
\[
\frac{CK^2}{n}
\leq\frac{CC_1}{C_2}
\leq\frac18\log 2
\leq\frac18\log(d).
\]
Consequently, we have
\[
\exp\!\left[
C\min\left\{\frac{K}{n^{1/4}},K\sqrt{\frac dn}\right\}
+\frac{CK^2}{n}
\right]
\leq d^{1/2}.
\]

Now \eqref{eq:parameter-choice} gives
$n\geq C_2\log^2(d)/\varepsilon^2$, while \eqref{eq:K-bounds} gives
$K^4\geq C_1^2d^2/\varepsilon^2$. Applying Proposition~\ref{prop:variance},
\begin{equation}
\begin{aligned}
\Var\sbra{\widehat S\mid\rho}
&\leq \frac{C\log^2(d)}{n}
 +\frac{C}{K^4}d^{1/2}\\
&\leq \frac{C\varepsilon^2}{C_2}
 +\frac{C\varepsilon^2}{C_1^2d^{3/2}}\\
&\leq \frac{\varepsilon^2}{128}+\frac{\varepsilon^2}{128}
=\frac{\varepsilon^2}{64},
\end{aligned}
\label{eq:final-variance}
\end{equation}
where the last line follows from the choices of $C_1$ and $C_2$ and $d\geq2$.

For the bias, Proposition~\ref{prop:bias} and the lower bound on $K$ in \eqref{eq:K-bounds} give
\[
\abs{\E\sbra{\widehat S\mid\rho}-\S\rbra{\rho}}
\leq\frac{2d}{K(K+1)}
\leq\frac{2d}{K^2}
\leq\frac{2\varepsilon}{C_1}
\leq\frac{\varepsilon}{4}.
\]
Finally, Chebyshev's inequality and \eqref{eq:final-variance} imply
\[
\Pr\!\left[
\abs{\widehat S-\E\sbra{\widehat S\mid\rho}}>\frac{\varepsilon}{2}
\,\middle|\,\rho
\right]
\leq
\frac{\varepsilon^2/64}{(\varepsilon/2)^2}
=\frac1{16}.
\]
On the complementary event,
\[
\abs{\widehat S-\S\rbra{\rho}}
\leq
\abs{\widehat S-\E\sbra{\widehat S\mid\rho}}
+
\abs{\E\sbra{\widehat S\mid\rho}-\S\rbra{\rho}}
\leq\frac{3\varepsilon}{4}<\varepsilon.
\]
Hence the success probability is at least $15/16>2/3$. 
\end{proof}

\section*{Acknowledgment}
The authors used Large Language Models as AI-assisted research and writing tools throughout the
preparation of this paper. The tools were used to help brainstorm ideas and explore proof
strategies. Portions of the paper text were redrafted or modified with AI assistance across all
sections. All final mathematical claims, algorithms, proofs, citations, and wording were reviewed,
edited, and validated by the authors. The authors assumes responsibility for all content of the paper.

\addcontentsline{toc}{section}{References}
\bibliographystyle{alphaurl}
\bibliography{main}

\appendix

\section{Proof of the entropy-polynomial bounds}
\label{app:entropy-polynomial}

We now prove Lemma~\ref{lem:entropy-polynomial}. 

\subsection{Derivative bounds}
\label{app:derivative-bounds}

Let $U_m$ be the Chebyshev polynomial of the second kind,
\[
U_m(\cos\theta)=\frac{\sin((m+1)\theta)}{\sin\theta}.
\]
The identities $T_k'=kU_{k-1}$ and
$(1-t^2)T_k''-tT_k'+k^2T_k=0$ imply, after differentiating the latter $r$ times,
\[
(1-t^2)T_k^{(r+2)}(t)
-(2r+1)tT_k^{(r+1)}(t)
+(k^2-r^2)T_k^{(r)}(t)=0.
\]
Evaluating at $t=1$ gives, for $1\leq r\leq k$,
\begin{equation}
\norm{T_k^{(r)}}_{\infty,[-1,1]}
=T_k^{(r)}(1)
=\frac{\prod_{\ell=0}^{r-1}(k^2-\ell^2)}{(2r-1)!!}
\leq \frac{k^{2r}}{(2r-1)!!}.
\label{eq:chebyshev-endpoint}
\end{equation}
In fact, note that
\[
U_{k-1}(\cos\theta)
=\sum_{\ell=0}^{k-1}\cos((k-1-2\ell)\theta),
\]
so $T_k'$ is a nonnegative linear combination of Chebyshev polynomials; repeated differentiation preserves this property.

We also need the weighted estimate
\begin{equation}
x^{r/2}\abs{T_k^{(r)}(1-x)}\leq(3k)^r,
\qquad 0\leq x\leq1,\quad1\leq r\leq k.
\label{eq:chebyshev-weighted}
\end{equation}
For $r=1$, the trigonometric formula gives the claim. If
$x\leq(r/k)^2$, \eqref{eq:chebyshev-endpoint} and
$(2r-1)!!\geq r!\geq(r/e)^r$ give the bound. If
$x>(r/k)^2$ and $r\geq2$, the differentiated Chebyshev equation gives, with
$v_r=x^{r/2}\abs{T_k^{(r)}(1-x)}$,
\[
v_r\leq2k\,v_{r-1}+k^2v_{r-2},
\]
and strong induction yields \eqref{eq:chebyshev-weighted}.

For $j\geq3$, differentiating \eqref{eq:entropy-polynomial} and using \eqref{eq:chebyshev-weighted} gives
\[
x^{j/2}\abs{p_K^{(j)}(x)}
\leq C^j\sum_{k=j}^K k^{j-3}
\leq C^jK^{j-2}.
\]
On the other hand, \eqref{eq:chebyshev-endpoint} gives
\[
\norm{p_K^{(j)}}_{\infty,[0,1]}
\leq \frac{C}{(2j-1)!!}\sum_{k=j}^Kk^{2j-3}
\leq (C/j)^jK^{2j-2}.
\]
Taking the smaller bound proves \eqref{eq:p-higher-derivative} for $j\geq3$.

For $j=2$, the identities
$2tU_{k-1}=U_k+U_{k-2}$ and $2T_k=U_k-U_{k-2}$ yield the telescoping formula
\begin{equation}
p_K''(x)
=\frac{U_{K-1}(t)/K+U_K(t)/(K+1)-1-t}{1-t^2},
\qquad t=1-x.
\label{eq:p-second-telescope}
\end{equation}
Since $\abs{U_m(t)}\leq m+1$ on $[-1,1]$,
\[
x\abs{p_K''(x)}\leq4,
\qquad 0<x\leq1.
\]
The endpoint formula also gives
$\norm{p_K''}_{\infty,[0,1]}\leq K^2$.
Therefore
\[
\abs{p_K''(x)}
\leq\frac{4}{\max\{x,4/K^2\}},
\]
which is \eqref{eq:p-higher-derivative} for $j=2$.

Finally, $\abs{p_K'(1)}<3$, and integrating the preceding weighted second-derivative bound from $x$ to $1$ gives
\[
\abs{p_K'(x)}\leq3+4\log(1/x),
\qquad0<x\leq1,
\]
proving \eqref{eq:p-first-derivative}.

\subsection{Expansion and uniform approximation}
\label{app:approximation}

The polynomials $p_K$ converge uniformly on $[0,2]$, because
\[
\abs*{\frac{2(T_k(1-x)-1)}{k(k^2-1)}}
\leq \frac{4}{k(k^2-1)}
\]
and the right-hand side is summable. Let $P$ denote the limit. Since $p_K(0)=0$, we have $P(0)=0$.

Rearranging \eqref{eq:p-second-telescope},
\[
p_K''(x)
=-\frac1x+
\frac{U_{K-1}(1-x)/K+U_K(1-x)/(K+1)}{x(2-x)}.
\]
For every compact subinterval of $(0,2)$, the second term tends uniformly to zero, so $p_K''(x)\to-1/x$. Moreover,
\[
p_K'(1)
=\log2-\frac32
+\frac12\sum_{m=1}^{\lfloor(K-1)/2\rfloor}
\frac{(-1)^m}{m(m+1)}
\longrightarrow -1.
\]
Taylor's formula around $1$ therefore implies
$P(x)=P(1)-x\log x$ for $0<x<2$. Continuity at $0$ and $P(0)=0$ force $P(1)=0$.

Finally, using $\abs{T_k(t)}\leq1$ for $t\in[-1,1]$,
\[
\sup_{0\leq x\leq2}\abs{p_K(x)+x\log x}
\leq4\sum_{k>K}\frac1{k(k^2-1)}
=\frac{2}{K(K+1)}.
\]
This proves \eqref{eq:p-zero} and \eqref{eq:p-approx}, completing the proof of Lemma~\ref{lem:entropy-polynomial}.

\end{document}